\documentclass[11pt]{article}

\usepackage{geometry}

\usepackage{enumitem}
\usepackage{amsmath}
\usepackage{graphicx}
\usepackage{amsthm}
\usepackage{amssymb}
\usepackage{hyperref} 
\usepackage{color}
\usepackage{calrsfs}
\usepackage{thmtools}
\usepackage{amssymb}

\author{Dhruv Rohatgi \\ MIT \\ drohatgi@mit.edu}

\newcommand{\smin}{\textsc{Sliding-Min}}
\newcommand{\smaj}{\textsc{Sliding-Majority}}
\newcommand{\ksmin}{\textsc{Sliding}-$l^\text{th}$-\textsc{Smallest}}

\newcommand\restr[2]{{
  \left.\kern-\nulldelimiterspace 
  #1 
  \vphantom{\big|} 
  \right|_{#2} 
  }}

\newtheorem{theorem}{Theorem}[section]
\newtheorem{lemma}[theorem]{Lemma}
\newtheorem*{claim}{Claim}
\newtheorem{proposition}[theorem]{Proposition} 
\newtheorem{corollary}[theorem]{Corollary} 

\theoremstyle{definition}
\newtheorem{definition}[theorem]{Definition}

\mathchardef\mhyphen="2D

\newtheoremstyle{break}
  {\topsep}{\topsep}%
  {}{}%
  {\bfseries}{}%
  {\newline}{}%
\theoremstyle{break}
\newtheorem*{algorithm}{Algorithm}

\theoremstyle{remark}

\title{Sliding window order statistics in sublinear space}

\begin{document}

\maketitle

\begin{abstract}
We extend the multi-pass streaming model to sliding window problems, and address the problem of computing order statistics on fixed-size sliding windows, in the multi-pass streaming model as well as the closely related communication complexity model. In the $2$-pass streaming model, we show that on input of length $N$ with values in range $[0,R]$ and a window of length $K$, sliding window minimums can be computed in $\widetilde{O}(\sqrt{N})$. We show that this is nearly optimal (for any constant number of passes) when $R \geq K$, but can be improved when $R = o(K)$ to $\widetilde{O}(\sqrt{NR/K})$. Furthermore, we show that there is an $(l+1)$-pass streaming algorithm which computes $l^\text{th}$-smallest elements in $\widetilde{O}(l^{3/2} \sqrt{N})$ space. In the communication complexity model, we describe a simple $\widetilde{O}(pN^{1/p})$ algorithm to compute minimums in $p$ rounds of communication for odd $p$, and a more involved algorithm which computes the $l^\text{th}$-smallest elements in $\widetilde{O}(pl^2 N^{1/(p-2l-1)})$ space. Finally, we prove that the majority statistic on boolean streams cannot be computed in sublinear space, implying that $l^\text{th}$-smallest elements cannot be computed in space both sublinear in $N$ and independent of $l$.
\end{abstract}

\section{Introduction}
\subsection{Background and Related Work}
The streaming model of computation has seen considerable and growing attention over the past few decades, driven by a wide variety of applications---from network monitoring to processing financial records to data mining---with at least one feature in common: data sets which are transient and therefore must be processed on arrival, or which are simply so massive that they cannot be stored in main memory.  Motivated by the former application of unbounded or possibly real-time data streams, much work has been done on ``one-pass streaming algorithms'', where access to the data is limited to a single linear scan, and efficiency is measured in terms of not only time but also space. Exact compution is typically infeasible in this model, for most nontrivial problems, so typically one must be content with computing an approximation, or (in light of the fact that the data stream may be unbounded in length) maintaining an approximate ``sketch'' of the prefix of the stream which has been processed. Work in this model ranges from estimating quantiles \cite{greenwald2001}, frequent elements \cite{charikar2004, cormode2008}, and norms \cite{kane2010}; to geometric problems such as clustering \cite{guha2000}; to graph problems such as property testing or estimating maximum matchings \cite{mcgregor2014}.

A distinct line of work has focused on the multi-pass streaming model, in which algorithms are restricted to a few passes through a finite-length data stream, and sublinear space. Here, the application is processing massive data sets which are stored in external memory: under this setting, linear scans are far more efficient than random-access I/O, and the data sets are too large to permit linear space consumption. Since Munro and Paterson's seminal paper on selection and sorting in this model \cite{munro1980}, a variety of problems have been studied in the context of multi-pass algorithms, and the tradeoff between space consumption and number of passes \cite{guha2008, chan2005, demaine2014, drineas2003, feigenbaum2005}.

Within the work on one-pass algorithms, there has been a growing body of research on the sliding window model, which addresses applications such as medical monitoring and financial monitoring (to name just a few), where the data stream has an intrinsic chronology, and it is desired that computation be done on ``recent'' elements of the data stream. Datar et al. introduced the model and considered the problems of approximately maintaining the sum and $\ell^p$ norm of the last $N$ elements in the data stream seen so far \cite{datar2002}. Since then, a variety of other statistics have been studied in the sliding window model \cite{babcock2003, arasu2004, braverman2007, braverman2010, braverman2012, braverman2013, crouch2013, braverman2015}. However, to our knowledge, all work in this model has been on one-pass approximation algorithms.

In this paper, we extend the sliding window model to allow multiple passes, in turn asking for exact algorithms rather than approximation algorithms. We concern ourselves with applications in which the data set is large but bounded (e.g. stored in external memory). It is not unreasonable to expect that many massive data sets may be intrinsically ``ordered''---for instance, data which was generated in real-time but is being processed later. Therefore it may be natural to wish to compute some statistic for every fixed-size window in the input. And one of the most natural problems is computing the \textsc{MIN/MAX} statistic, or its generalization, the $k^\text{th}$-smallest element in each window. As further motivation, Datar et al. showed that no single-pass algorithm can maintain MIN or MAX in sublinear memory \cite{datar2002}. Allowing a constant-factor approximation, the saving grace of most statistics, only permits a constant factor decrease in memory. Hence, it is natural to rather ask whether allowing multiple passes permits a sublinear-space algorithm.

\subsection{Models and Contributions}

We must define the multi-pass streaming model for sliding window problems. The basic problem is to compute some statistic for each fixed-length window in the input. In the previously studied single-pass streaming model, it is required that each statistic is output immediately after reading the last element in the corresponding window. However, in the multi-pass streaming model this constraint seems to be unnaturally strong. The intended application is processing data which is stored in external memory, for which I/O consisting of linear scans requires the least overhead. Since the output of a sliding window algorithm is linear in the size of the input, it too must be stored in external memory. Therefore the efficiency of printing the output can be measured by the number of passes which the algorithm makes through the output array. With this motivation, we define the multi-pass sliding window model to have a write-only \textit{output} data stream (introduced, to our knowledge, in \cite{francois2014}, though not in the context of sliding windows). For an input of length $N$ with a sliding window of length $K$, the output stream has length $N-K+1$, and an algorithm produces output by making one or several linear scans through the output stream. To avoid defining a mechanism for overwriting, we require that every position in the output stream is written to exactly once. We do not require that the output stream is in any way ``synced'' to the input stream.

For example, we say that a $p$-pass algorithm makes one output pass if the relative order of the output is correct: the first value output by the algorithm must be the statistic for the leftmost window in the input, and so forth. And for a $p$-pass algorithm to use $p$ output passes, a sufficient though not necessary criterion is that in any fixed pass, the window answers are output in order. That is, for any fixed pass, if the answer for $[i, i+K)$ is output during the pass, and the answer for $[j, j+K)$ is output later in the same pass, then $i < j$.

All of our results in the multi-pass streaming model will be for algorithms restricted to a single output pass, unless we explicitly state otherwise.

There is a natural adaptation of the multi-pass streaming model to the two-party communication complexity model, made precise as follows. One party (conventionally known as Alice) is given the first half of the input array, and another party (known as Bob) is given the second half. During the $p$ rounds of communication, Alice and Bob must output the sliding window minimums for all windows to a shared, write-only output stream. Alice may output some, and Bob may output some, but the number of output passes is measured as in the streaming problem. As a technicality, neither Alice nor Bob have read-access to their shared position in the output stream. All of our communication complexity results will be for the single output pass model.

In both models, we will assume that the algorithm has access to the stream size $N$, the window size $K$, and a bound $R$ on the maximum integer value present in the stream, beforehand.

The main problem we consider is computing the $l^\text{th}$ smallest element on all windows of some fixed length. We call this problem \ksmin{}; in the special case of $l=1$, we refer to the problem as \smin{}. The trivial algorithm for \ksmin{} uses $\tilde{O}(K)$ space, where $K$ is the window length. Since $K$ can be linear in $N$, we seek a uniform improvement in complexity, over all $K$. Our first main result is an algorithm which provides this improvement:

\theoremstyle{theorem}
\newtheorem*{thm:ksmin}{Theorem \ref{thm:ksmin}}
\begin{thm:ksmin}
Let $l > 0$ be an integer. Then there is a streaming algorithm for \ksmin{} which uses $O(l)$ input and output passes and $\widetilde{O}(l^{3/2}\sqrt{N})$ space.
\end{thm:ksmin}

We then provide a corresponding improvement in the communication complexity setting:

\newtheorem*{ccgen}{Theorem \ref{ccgen}}
\begin{ccgen}
Let $p > 0$ be an integer. There is a $p$-round communication complexity algorithm for \ksmin{} which uses $\widetilde{O}(pl^2 N^{1/(p-2l-1)})$ communication bits.
\end{ccgen}

Theorem~\ref{ccgen} implies that arbitrarily efficient communication algorithms for \smin{} exist, using varying constant numbers of input rounds and a single output pass. In contrast, Theorem~\ref{thm:ksmin} implies that we can solve \smin{} in the streaming model using two input and one output pass, and $\widetilde{O}(\sqrt{N})$ space. But it is not at all clear how to trade passes for space. In fact, it turns out that this disparity is intrinsic, as shown by the following theorem, which describes a nearly tight lower bound for \smin{} in the streaming setting, for any constant number of input passes.

\newtheorem*{thm:lbound}{Theorem \ref{thm:lbound}}
\begin{thm:lbound}
Any algorithm for \smin{} which uses a constant number of input passes and a single output pass has a worst-case space complexity (over all $K$) of $\Omega(\sqrt{N})$.
\end{thm:lbound}

And we partially resolve the complexity of the sliding window majority problem \smaj{} in our new model. Observe that this is the special case of \ksmin{} in which $l = N/2$ and the maximum input value is $R = 1$. It is also a special case of basic statistics such as sum, median, and mode. Unfortunately, under the assumption of single output pass, there is no sublinear space algorithm:

\newtheorem*{thm:majlbound}{Theorem \ref{thm:majlbound}}
\begin{thm:majlbound}
Any algorithm for \smaj{} which uses a constant number of input passes and a single output pass has a worst-case space complexity of $\Omega(N)$.
\end{thm:majlbound}

To achieve the stated streaming upper bounds, our key technical tool is, informally, some notion of monotonicity: the index of the smallest element in a sliding window is non-decreasing as the window indices increase, and this allows interpolation of the minima indices after computing them for a sparse subset of the windows. For the general problem of $l^\text{th}$-smallest elements, it is necessary to develop a weaker invariant of ``near-monotonicity''. Our communication complexity results use these tools to split the input array by a small number of ``critical'' indices, and then solve each subarray independently.

For our lower bounds, we construct information theoretic proofs in the communication complexity model, and apply these to the streaming model. This second step is not entirely trivial, since we wish to prove bounds for multi-pass streaming algorithms, whereas the most convenient communication lower bounds apply to the one-round communication model. Furthermore, as seen in our results, it is in fact possible to solve \smin{} with $o(\sqrt{N})$ memory in the communication complexity model, so a direct application of the communication bound could not possibly work. To address the former issue, we simply localize to a single pass in which significant work was performed. We circumvent the latter barrier with techniques that could generalize to prove sublinear lower bounds for other problems in this model: we divide the input stream into blocks, show that the communication complexity lower bound can be applied independently to many blocks, and compound the resulting bounds.

\subsection{Organization}
The remainder of the paper is organized as follows. In Section~\ref{sec:prelim}, we discuss some notation and preliminaries that will be useful throughout the paper. Section~\ref{sec:min} contains sublinear-space streaming and communication algorithms for \smin{}. In section~\ref{sec:lthsmallest} we present our main algorithmic results, generalizing our $l=1$ algorithms to arbitrary $l$ in both models. In section~\ref{sec:lowerbounds}, we prove a nearly tight lower bound for \smin{}. Finally, we show in Section~\ref{sec:maj} that the trivial linear space algorithm for \smaj{} is nearly optimal.

\section{Preliminaries}\label{sec:prelim}

Throughout this paper, unless we specify otherwise, we will refer to $A$ as our input array of nonnegative integers, $N$ as the number of elements in the input, $K$ as the size of the sliding window, and $R$ as the upper bound on element values in the input. And we use the notation $[a,b]$ to denote the set of integers between $a$ and $b$ inclusive. 

Moving to notation more specific to this paper, we will supplement our streaming algorithms with pseudocode, in which we let $\textsc{Process}(p,j)$ denote the algorithm for processing the $j^\text{th}$ element of the input stream in pass $p$. And when we refer to the ``$p$-pass streaming model'', we mean the multi-pass streaming model restricted to $p$ input passes and one output pass: if several output passes are required, we will make note of this fact.

Finally, we must describe a data structure that will be a useful building block and black box in one of our algorithms. This data structure is typically known as a monotonic queue. \footnote{This data structure is widely known in the world of competitive programming, and has probably been used in algorithms research in the past, but we do not know if it is attributable to a particular paper in the literature.} It is essentially a queue which maintains the invariant that its elements are in strictly increasing order. The following lemma summarizes the relevant properties:

\begin{lemma}\label{lemma:mqueue}
Given a stream of integers $a_1, a_2, \dots$, there is a data structure which supports the operations $\mathrm{insert}(i, a_i)$, $\mathrm{get \mhyphen front}$, and $\mathrm{delete \mhyphen front}$. The inserted indices must be in strictly increasing order. The operation $\mathrm{get \mhyphen front}$ returns the index and value of the smallest element (by value) which has larger index than that of the most recently deleted element; the operation $\mathrm{delete \mhyphen front}$ deletes that element. All operations require amortized constant time, and space usage is at all times linear in the number of insertions minus deletions, ignoring logarithmic factors.
\end{lemma}

In the standard model of computation, the monotonic queue is the canonical data structure used to compute sliding window minimums in $O(N)$ time. As input elements are passed by the leading end of the sliding window, they are inserted. As input elements are passed by the lagging end of the sliding window, they are deleted from the front of the queue if necessary. The minimum element in a sliding window is simply the minimum element in the queue. The amortized time complexity is $O(N)$. If the input elements are integers in $[0,R]$, then the space complexity is $O(\min(K,R) (\log N + \log R))$.

It's worth noting why we should not be content with the monotonic queue algorithm in the streaming setting. For small $K$, this algorithm is of course efficient. However, in the ``worst case'', when $K = \Theta(N)$, it can be quite bad: the algorithm uses linear space. In this paper, we'll look at improving this worst-case performance, to find algorithms that are space-efficient for arbitrary $K$.

\section{Algorithms for \smin{}}\label{sec:min}

To introduce some of the key ideas for our algorithms in a simpler setting, we start with the special case of \ksmin{} in which $l = 1$. Datar et al. showed that in the single-pass streaming model, \smin{} requires $\Theta(K)$ space \cite{datar2002}. That is, the trivial algorithm which stores the last $K$ elements is essentially optimal. However, it is possible to solve \smin{} more efficiently with more passes. In section~\ref{sec:basicalg}, we show that there is a $2$-pass streaming algorithm for \smin{} using $O(\sqrt{N} (\log R + \log N))$ space and $O(N)$ time. In section~\ref{sec:smallint}, we describe a $2$-pass streaming algorithm using $O(\sqrt{NR/K} (\log R + \log N))$ space and $O(N)$ time, which is more space efficient if $R = o(K)$.

As is to be expected, our results for the communication complexity model are stronger. In section~\ref{sec:basiccc}, we show that for any odd $p$, in the $p$-round communication complexity model, which is a relaxation of the $\lfloor p/2 \rfloor$ streaming model, \smin{} can be solved with $O(pN^{1/p} \log R)$ bits of communication.

\subsection{A Two-Pass Streaming Algorithm}\label{sec:basicalg}

We start by describing a streaming algorithm that solves \smin{} in sublinear space with two passes through the input. Crucial to the algorithm is monotonicity, which is a motif that will recur in various forms throughout the paper. Fix some input $A$ of length $N$ with window size $K$. For each index $i \in [0, N-K]$, let $\text{low}(i)$ be the minimum value in window $[i,i+K-1]$. Let $f(i)$ be the earliest index in $[i,i+K-1]$ achieving value $\text{low}(i)$.

Note that $f$ is non-decreasing. Therefore computing $f$ on a sparse set of indices in one pass provides some information that may be useful for computing $f$ on all indices in the second pass. This observation is the heart of the following algorithm. In the first pass, we compute $f(i\sqrt{N})$ and $\text{low}(i\sqrt{N})$ for $0 \leq i < \sqrt{N}$. Now we observe that $f$ can now be computed in blocks of $\sqrt{N}$. If $f(i\sqrt{N}) < j < f((i+1)\sqrt{N})$, then $A_j$ can only be the minimum for windows $[i',i'+K)$ with $i\sqrt{N} < i' < (i+1)\sqrt{N}$, so only $\sqrt{N}$ counters need to be maintained at any time. Hence, in the second pass, we scan through the array, maintaining minima for the first $\sqrt{N}$ windows, until index $f(\sqrt{N})$ is reached. Output the first $\sqrt{N}$ minima, initialize minima for the next $\sqrt{N}$ windows, and repeat.

\begin{figure}
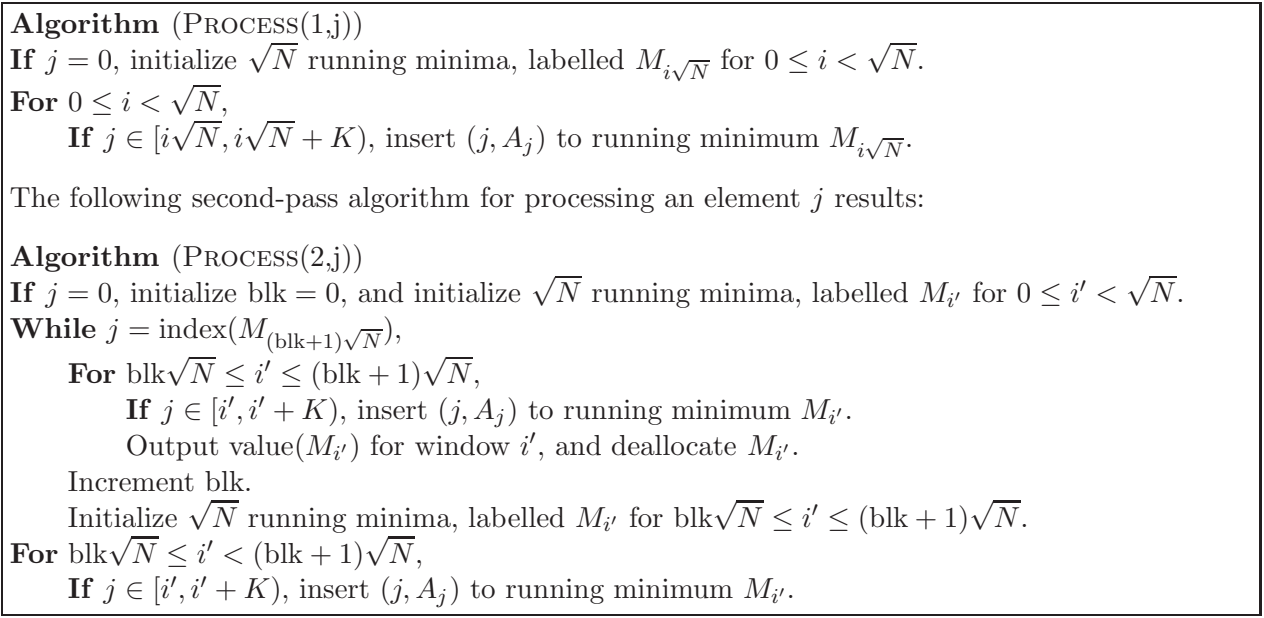

\fbox{\begin{minipage}{\textwidth}
\begin{algorithm}[\textsc{Process}(1,j)]
{\parindent0pt
\textbf{If} $j = 0$, initialize $\sqrt{N}$ running minima, labelled $M_{i\sqrt{N}}$ for $0 \leq i < \sqrt{N}$.

\textbf{For} $0 \leq i < \sqrt{N}$,

\qquad \textbf{If} $j \in [i\sqrt{N},i\sqrt{N}+K)$, insert $(j,A_j)$ to running minimum $M_{i\sqrt{N}}$.
}
\end{algorithm}

 The following second-pass algorithm for processing an element $j$ results:

\begin{algorithm}[\textsc{Process}(2,j)]
{\parindent0pt
\textbf{If} $j = 0$, initialize $\text{blk} = 0$, and initialize $\sqrt{N}$ running minima, labelled $M_{i'}$ for $0 \leq i' < \sqrt{N}$.

\textbf{While} $j = \text{index}(M_{(\text{blk}+1)\sqrt{N}})$,

\qquad \textbf{For} $\text{blk}\sqrt{N} \leq i' \leq (\text{blk}+1)\sqrt{N}$,

\qquad \qquad \textbf{If} $j \in [i',i'+K)$, insert $(j,A_j)$ to running minimum $M_{i'}$.

\qquad \qquad Output $\text{value}(M_{i'})$ for window $i'$, and deallocate $M_{i'}$.

\qquad Increment $\text{blk}$.

\qquad Initialize $\sqrt{N}$ running minima, labelled $M_{i'}$ for $\text{blk}\sqrt{N} \leq i' \leq (\text{blk}+1)\sqrt{N}$.

\textbf{For} $\text{blk}\sqrt{N} \leq i' < (\text{blk}+1)\sqrt{N}$,

\qquad \textbf{If} $j \in [i',i'+K)$, insert $(j,A_j)$ to running minimum $M_{i'}$.
}
\end{algorithm}

\end{minipage}}

\caption{The two-pass algorithm for \smin{}}
\label{figure:twopass}
\end{figure}

This leads to the pseudocode in Figure~\ref{figure:twopass}. Analyzing the space complexity yields the following result:

\begin{proposition}\label{twopass}
In the $2$-pass streaming model, \smin{} can be solved in $O(\sqrt{N} (\log R + \log N))$ space.
\end{proposition}

\begin{proof}
Correctness has essentially been shown already: consider any block index $i$ and index $j \in [i\sqrt{N},(i+1)\sqrt{N})$. Then the running minimum for the window starting at $j$ is maintained from when index $f(i\sqrt{N})$ is reached in the stream, to when index $f((i+1)\sqrt{N})$ is reached, inclusive. By the monotonicity observation, this is sufficient: by the time the running minimum is output, it is equal to $\text{low}(j)$, as desired.

At any time during either pass, only $O(\sqrt{N})$ variables are in memory, each taking a value in $[0,R]$. Therefore the space complexity of each pass is $\sqrt{N}\log R$.
\end{proof}

Since each stream element may trigger updates to $\sqrt{N}$ running minimums, the time complexity of the above algorithm is $O(\sqrt{N})$ per stream element. It is not difficult to improve the overall time complexity of the above algorithm to $O(N)$ using monotonic queues, but we will not elaborate further on time efficiency.

\subsection{A Small-Integer Algorithm}\label{sec:smallint}

When the size of stream elements is small (in particular, $R = o(K)$), it is possible to improve upon the above $\Theta(N^{1/2})$ space complexity. As the following lemma formalizes, if the sliding window minimum is restricted to a small set of values, then it cannot change ``too often." Like before, we let $\text{low}(i)$ refer to the minimum of window $[i,i+K-1]$.

\begin{lemma}\label{boundedchange}
Let $A$ be an array of $N$ integers in $[0,R]$, and let the window size be $K$. Then there are $O(NR/K)$ indices $i \in [0,N-K]$ at which $$\text{low}(i) \neq \text{low}(i+1).$$
\end{lemma}

\begin{proof}
Suppose $\text{low}(i) > \text{low}(i+1)$ for some index $i$. Then $\text{low}(i+1) = A_{i+K}$. But for each index $j \in [i+1,i+K]$, the value $A_{i+K}$ is in the window $[j,j+K-1]$. Hence $\text{low}(j) \leq \text{low}(i+1)$ for $j \in [i+1,i+K]$. So after any decrease, $\text{low}(i)$ cannot increase in the next $K-1$ windows.

Consider any $K$ consecutive windows. By the above fact, every increase must occur before every decrease. But $\text{low}(i)$ can increase at most $R$ times without any interleaved decreases, and it can decrease at most $R$ times without any interleaved increases. So $\text{low}(i)$ changes $O(R)$ times among these $K$ windows. It follows that $\text{low}(i)$ can change only $O(NR/K)$ times in total.
\end{proof}

To use this, we divide the interval $[0,R]$ into $P$ equally sized buckets, where $P$ is a constant we will pick later. In the first pass, we use a monotonic queue to compute for each window which bucket its minimum is in. To save space, we only store the indices where the minimum changes (and the new minima). Now let $i_1,\dots,i_k$ be the indices at which the minimum increases; these can be computed from the list $L$. For ease of notation, let $i_{k+1} = N$. In the second pass, we divide the input by these indices, and maintain a monotonic queue for each subarray, while ensuring that the queue only ever contains elements from a single bucket.

Now let $i_1,\dots,i_k$ be the indices at which the minimum increases; these can be computed from the list $L$. For ease of notation, let $i_{k+1} = N$. In the second pass, we divide the input by these indices, and maintain a monotonic queue for each subarray. The key is that we can ensure the queue only ever contains elements from a single bucket. More formally: for each $k$ sequentially, read the input stream until index $i_k$, maintaining a monotonic queue with window $K$, inserting all elements which do not lie in larger buckets than the front of the queue. After processing any element $A_j$ with $j-K+1 \geq i_{k-1}$, output the front of the queue as the answer for window $[j-K+1, j]$. And before processing element $A_{i_k}$, output the front of the queue as the answer for each remaining window $[i_\text{out},i_\text{out}+K)$ until the output pointer $i_\text{out}$ reaches $i_k$, deleting the front of the queue each time it exits the current output window. Then increment $k$, proceeding as above until the entire stream is processed. See Figure~\ref{figure:smallintcode} for the pseudocode.

\begin{figure}
\fbox{\begin{minipage}{\textwidth}

\begin{algorithm}[\textsc{Process}$(1,j)$]
{\parindent0pt
\textbf{If} $j = 0$, initialize an empty monotonic queue MQ, and an empty list $L$.

\textbf{Insert} $\left (j, \left \lfloor \frac{PA_j}{R} \right \rfloor \right)$ to MQ.

\textbf{While} $\text{index}(\text{MQ}.\text{front}) \leq j-K$, \textbf{delete} front of MQ.

\textbf{If} $L$ is empty or $\text{value}(\text{MQ}.\text{front}) \neq \text{value}(L.\text{last})$, append $\text{MQ}.\text{front}$ to $L$.
}
\end{algorithm}

\begin{algorithm}[\textsc{Process}$(2,j)$]
{\parindent0pt
\textbf{If} $j = 0$, initialize an empty monotonic queue MQ, and initialize $\text{blk} = 0$.

\textbf{While} $\text{index}(\text{MQ}.\text{front}) \leq j-K$, \textbf{delete} front of MQ.

\textbf{If} $\left \lfloor \frac{PA_j}{R} \right \rfloor \leq \left \lfloor \frac{P \cdot \text{value}(\text{MQ}.\text{front})}{R} \right \rfloor$, \textbf{insert} $\left (j, A_j \right)$ to MQ.

\textbf{If} $j - K + 1 \geq i_\text{blk}$, output $\text{value}(\text{MQ}.\text{front})$ for window $[j-K+1, j]$.

\textbf{If} $j = i_\text{blk+1} - 1$,

\qquad \textbf{For} $\max(i_\text{blk}, j - K + 2) \leq i < i_\text{blk+1}$,

\qquad \qquad \textbf{While} $\text{index}(\text{MQ}.\text{front}) < i$, \textbf{delete} front of MQ.

\qquad \qquad Output $\text{value}(\text{MQ}.\text{front})$ for window $[i,i+K)$.

\qquad Increment blk.
}
\end{algorithm}

\end{minipage}}

\caption{The small-integer algorithm for \smin{}}
\label{figure:smallintcode}
\end{figure}

\begin{claim}
The above algorithm is correct.
\end{claim}

\begin{proof}
Correctness of the first pass follows from Lemma~\ref{lemma:mqueue}. For the second pass, two issues must be argued. First, we must show that after reading element $i_\text{blk+1}-1$, we can read off $\text{low}(i)$ for each $i < i_\text{blk+1}$ from the queue. This holds since $\text{low}(i_\text{blk+1}-1)$ is in a smaller bucket than $\text{low}(i_\text{blk+1})$, so $A_{i_\text{blk+1}-1}$ is smaller than every element in range $[i_\text{blk+1},i_\text{blk+1}+K)$.

Second, we must show that it's not necessary to insert an element $A_j$ to the queue if $A_j$ lies in a bucket larger than the bucket of the front of the queue. Let $j'$ be the index of the element at front of the queue, and suppose $A_j$ lies in a larger bucket than $A_{j'}$. Observe that $j' > j - K$ and $j' \geq i_c$, where $c$ is such that $j \in [i_c, i_{c+1})$. Consider the window starting at $\max(j-K+1, i_c)$. Its minimum is at most $A_{j'}$, so the minimum lies in a smaller bucket than $A_j$. But the minimum's bucket only increases at $i_1,\dots,i_k$, so for any window $[i,i+K)$ with $\max(j-K+1,i_c) \leq i \leq j$, the minimum lies in a smaller bucket than $A_j$. Therefore $A_j$ is the minimum element of no window containing it, so inserting $A_j$ to the monotonic queue is unnecessary.
\end{proof}

\begin{proposition}
In the $2$-pass streaming model, \smin{} can be solved in $O(\sqrt{NR/K} (\log R + \log N))$ space.
\end{proposition}

\begin{proof}
The first pass uses $O(P (\log P + \log N))$ space for the monotonic queue, and by Lemma~\ref{boundedchange}, it uses $O((NP/K) (\log P + \log N))$ space to store the approximate minima. 

For the second pass, we must bound the length of the monotonic queue. Suppose $A_i$ has just been inserted into the queue, for some arbitrary index $i$. Then $A_i$ is not in a larger bucket than the front element of the queue. But monotonic queues are kept sorted in increasing order, so no element of the queue is in a larger bucket. And since the queue is monotonic, no element can be in a smaller bucket. Thus all elements of the queue are in the same bucket, so the queue has size at most $R/P$, the size of the bucket. 

We conclude that the space used by the second pass is $$O((R/P + NP/K)(\log P + \log N)),$$ to store the queue and the information from the first pass. The space used by the first pass is no greater. The optimum choice for $P$ is $\sqrt{RK/N}$, yielding space complexity $O(\sqrt{NR/K}(\log R + \log N)$.
\end{proof}

\subsection{A Two-Party Communication Algorithm}\label{sec:basiccc}

The streaming algorithm for \smin{} presented in Theorem~\ref{twopass} was in fact motivated by a similar algorithm for the communication complexity analogue. The key idea is the same: the index of the minimum is non-decreasing. In the two-party communication model, the implication is that there exists some index $i_\text{last}$ such that the minimum of window $[i,i+K)$ lies in Alice's input if and only if $i \leq i_\text{last}$. Thus for $i \leq i_{\textrm{last}}$, Alice can determine the minimum with no further information from Bob. And for $i > i_{\textrm{last}}$, Bob can determine it with no further information from Alice. So essentially, the task of Alice and Bob is to determine $i_{\textrm{last}}$ efficiently. 

In the following $3$-round algorithm, Alice and Bob pinpoint $i_\text{last}$ through a $D$-ary search, where $D = N^{1/3}$. Let $a_i$ denote the minimum of the last $i$ elements of Alice's input, and let $b_j$ denote the minimum of the first $j$ elements of Bob's input.

\begin{algorithm}
{\parindent0pt
1. Alice sends to Bob $a_{iD^2}$ for $1 \leq i \leq D/2$.

2. Bob determines the unique $i'$ such that $N/2 - i_\text{last}$ is in the range $[i'D^2, (i'+1)D^2)$. Next, Bob sends to Alice $i'$ along with $b_{K - i'D^2 - jD}$ for each $0 \leq j \leq D$.

3. Alice determines the unique $j'$ such that $N/2 - i_\text{last}$ is in the range $$[i_1, i_2] = [i'D^2 + j'D, i'D^2 + (j'+1)D].$$ Alice outputs the minimums for all windows $[i, N/2)$ such that $i \leq i_1$. Then she sends $j'$ to Bob, along with each $a_i$ for $i_1 \leq i \leq i_2$.

4. Bob outputs the minimums for all windows $[N/2, i+K)$ such that $i > i_1$.
}
\end{algorithm}

The feasibility of the algorithm follows from the fact that the minimum of a window $[i,i+K)$ is equal to $\min(a_{N/2-i}, b_{i+K-N/2})$. So for instance, when Alice sends $a_{iD^2}$, Bob can determine whether $i_\text{last} \leq N/2 - iD^2$ or not. The communication used is $O(D \log R + \log N)$, and it's clear that the algorithm can be extended to any odd number of rounds. The below proposition follows.

\begin{proposition}\label{cc3}
For any odd $p$, in the $p$-round communication complexity model, \smin{} can be solved using $O(pN^{1/p} \log R)$ bits of communication.
\end{proposition}

\section{Algorithms for \ksmin{}}\label{sec:lthsmallest}

Unfortunately, the arguments of Theorem~\ref{twopass} and Proposition~\ref{cc3} do not directly generalize to \ksmin{} for $l>1$, since the indices of the $l^\text{th}$ smallest elements in successive windows are not necessarily non-decreasing. However, we can prove a weaker cousin of monotonicity for $l^\text{th}$ smallest elements, that still yields nontrivial algorithms in both the streaming model and communication model.

\subsection{Generalizing Monotonicity to \ksmin{}}\label{sec:gmonotonicity}

In this subsection, we'll develop the ``near-monotonicity'' results which we need for our \ksmin{} algorithms. A few definitions are needed first.

\begin{definition}
Given an array $A$ of $M$ integers indexed $1 \dots M$, let the \textit{rank} of the $i^\text{th}$ element $A_i$ be the number of elements $A_j$ such that either $A_i < A_j$ or $A_i = A_j$ and $j \leq i$.
\end{definition}
That is, the rank of $A_i$ is the $1$-index of $A_i$ when $A$ is sorted and ties are broken by location in the original array.

Fix an array $A$ of $N$ integers, and fix some window size $K$.

\begin{definition}
For any indices $i,j$ into $A$ with $i \leq j$, let $s_l([i,j])$ be the index $i'$ such that $A_{i'}$ has rank $l$ in $A_{[i,j]}$. Furthermore, let the \textit{$l$-lowset} of window $[i,j]$, denoted $S_l([i,j])$, be the set of indices $i' \in [i,j]$ such that the rank of $A_{i'}$ in $A_{[i,j]}$ is at most  $l$. That is, $$S_l([i,j]) = \{s_1([i,j]), \dots, s_l([i,j])\}.$$
\end{definition}


Suppose we sort the $l$-lowset of a window $[i,j]$ from smallest index to largest index (rather than from smallest rank to largest rank). The following lemma states that under this sorting, the indices of the $l$-lowsets of sliding windows $[i,i+K)$ are elementwise non-decreasing.

\begin{lemma}\label{nondec}
For any fixed $l$, the sets $S_l([i,i+K-1])$ are non-decreasing as $i$ increases. That is, for each $i$, the $k^\text{th}$ smallest index in set $S_l([i,i+K-1])$ is less than or equal to the $k^\text{th}$ smallest index in set $S_l([i+1,i+K])$.
\end{lemma}

\begin{proof}
Consider some window $[i,j]$. We seek to show that elements in the unordered set of indices $S_l([i,j])$ do not decrease when the window is shifted to $[i+1,j+1]$.

Consider the addition of $A_{j+1}$ to the window. If the rank of $A_{j+1}$ in $[i,j+1]$ is greater than $l$, then the set does not change at all. Suppose the rank is $l_{j+1} \leq l$. Then the indices of the smallest $l_{j+1}-1$ elements do not change. But the index of the element with rank $l_{j+1}$ is $j+1$, so $j+1$ is inserted into the set. For each $l' \in [l_{j+1}+1, l]$, the index of the element with rank $l'$ in $A_{[i,j+1]}$ is the index of the element with rank $l'-1$ in $A_{[i,j]}$, which is in $S_l([i,j])$ already. Finally, the element with rank $l$ in $A_{[i,j]}$ has rank $l+1$ in $A_{[i,j+1]}$, so it is removed from the set. In total, some index in $[i,j]$ is removed from the set, but $j+1$ is added to the set. Hence, the set does not decrease when $A_{j+1}$ is added to the window.

Now consider the removal of $A_i$ from the window $[i,j+1]$. If the rank of $A_i$ is greater than $l$, then $i$ is not in $S_l([i,j+1])$, so the set does not change. Otherwise, by analogous logic to the first step, index $i$ is removed from the set and some index in $[i+1,j+1]$ is added to the set. Hence the set does not decrease, and the desired result follows.
\end{proof}

Of course, the above lemma does not directly say anything about the motion of $s_l([i,i+K))$. In particular, we want to be able to detect when it crosses a given value. Some more notation must be defined.

\begin{definition}
For each index $i$ define the $l,b$-\textit{configuration} of window $[i,j]$ as the set of ranks $l' \in [1,l]$ such that the index of the element with rank $l'$ in $[i,j]$ is greater than $b$. Denote this set by $\text{Conf}_{l,b}([i,j])$.
\end{definition}

Now fix some rank $l$.

\begin{definition}
Let $S_{b-\text{cross}}$ be the set of indices $i$ so that $l$ is not in the $l,b$-configuration of window $[i-1,i+K-2]$ but is in the $l,b$-configuration of window $[i,i+K-1]$, or $l$ is in the $l,b$-configuration of window $[i-1,i+K-2]$ but not in the $l,b$-configuration of window $[i,i+K-1]$.
\end{definition}

The observation allowing the elements of $S_{b-\text{cross}}$ to be found is consideration of the following non-decreasing function.

\begin{definition}
Define for every index $b$ the function
$$g_{b}(i) = |\text{Conf}_{l,b}([i,i+K-1])| + |\text{Conf}_{l-1,b}([i,i+K-1])|$$ for all indices $i$.
\end{definition}

We have a key result:

\begin{lemma}\label{gmonotonicity}
Let $b$ be any index. Then for any index $i$, we have $g_b(i-1) \leq g_b(i)$. Furthermore, if $i \in S_{b-\text{cross}}$ then $g_b(i-1) < g_b(i)$.
\end{lemma}

\begin{proof}
By Lemma~\ref{nondec}, we know that the set $S_l([i,i+K-1])$ is elementwise greater than or equal to the set $S_l([i-1,i+K-2])$.  Hence the number of elements in $S_l([i,i+K-1])$ which exceed $b$ is at least the number of elements in $S_l([i-1,i+K-2])$ which exceed $b$. But this implies that $$|\text{Conf}_{l,b}([i,i+K-1])| \geq |\text{Conf}_{l,b}([i-1,i+K-2])|.$$ An identical argument holds for $\text{Conf}_{l-1,b}$, so we conclude that $g_b$ is non-decreasing.

For the second claim, there are two cases. In the first case, suppose that $l$ is not in $\text{Conf}_{l,b}([i-1,i+K-2])$ but is in $\text{Conf}_{l,b}([i,i+K-1])$. Then
\begin{align*}
|\text{Conf}_{l,b}([i,i+K-1])|
&= 1 + |\text{Conf}_{l-1,b}([i,i+K-1])|\\
&\geq 1 + |\text{Conf}_{l-1,b}([i-1,i+K-2])| \\
&= 1 + |\text{Conf}_{l,b}([i-1,i+K-2])|.
\end{align*}
Since $\text{Conf}_{l-1,b}$ is non-decreasing, this difference of $1$ suffices to show that $g_b(i) > g_b(i-1)$.

In the second case, $l$ is in $\text{Conf}_{l,b}([i-1,i+K-2])$ but is not in $\text{Conf}_{l,b}([i,i+K-1])$. Then
\begin{align*}
|\text{Conf}_{l-1,b}([i,i+K-1])|
&= |\text{Conf}_{l,b}([i,i+K-1])| \\
&\geq |\text{Conf}_{l,b}([i-1,i+K-2])| \\
&= 1 + |\text{Conf}_{l-1,b}([i-1,i+K-2])|.
\end{align*}
Combining the resulting strict inequality with the monotonicity of $\text{Conf}_{l,b}$ again yields $g_b(i) > g_b(i-1)$.
\end{proof}

On a side-note, it immediately follows from Lemma~\ref{gmonotonicity} and the bound $g_b(i) < 2l$ that the index with rank $l$ does not cross the boundary between any two consecutive elements indexed $b$ and $b+1$ very often, although this observation in itself (the direct generalization of $l=1$ monotonicity) is not sufficient to permit an efficient streaming or communication algorithm.

\begin{corollary}
For any index $b$ and any rank $l$, we have $|S_{b-\text{cross}}| < 2l$.
\end{corollary}

\subsection{An $(l+1)$-pass streaming algorithm}\label{sec:lstream}

Now we generalize our $2$-pass streaming algorithm for \smin{} to an $(l+1)$ input pass, $l$ output pass streaming algorithm for \ksmin{}. The first pass is analogous to the original first pass: compute the $l$-lowsets of $\sqrt{lN}$ evenly spaced $K$-length windows. The subsequent $l$ passes all perform identical computation; they only differ in the sets of windows for which output is produced in that pass. The algorithm will have to use more than just a single pass through the output stream, but by dividing the output in the following manner, we'll be able to show that the algorithm uses only $l$ output passes.

Let $B = \sqrt{N/L}$. For each $i$, call the range of windows starting between $iB$ and $(i+1)B$ a ``block.'' For each block $i$, we'll construct an interval $[a_i, b_i]$ such that $s_l([j, j+K)) \in [a_i, b_i]$ for all $iB \leq j \leq (i+1)B$, and we know its rank in the interval. If we can construct intervals small enough that any element is contained in only $O(l)$ intervals, then we'll be able to achieve an $O(lN/B + l^2B)$ space bound.

The following algorithm will make use of an ``$l$-capped ordered dictionary'' data structure which maintains, in order, the $l$ smallest values and associated indices inserted into the data structure, and has an ``insert'' operation. We refer to the $i^\text{th}$ smallest element of dictionary $D$ as $D(i)$.

In the first pass, maintain $N/B$ $l$-capped dictionaries, one for each window $[iB, iB+K)$ where $0 \leq i < N/B$. So after the first pass, dictionary $L_{iB}$ will store the $l$ smallest element/index pairs in window $[iB, iB+K)$. Now we describe the procedure for pass $p$, where $2 \leq p \leq l+1$. At the beginning of pass $p$, for each $i$, use $L_{iB}$ to compute the maximal interval $[a_i, b_i]$ around $L_{iB}(l)$ such that $g_j(iB) < g_j((i+1)B)$ for each $j \in [a_i, b_i)$. Also compute a flag array $O$, which indicates the blocks for which output is produced in this pass. Set $O(i) = 1$ if $b_i$ is the $(p-1)^\text{th}$ smallest element in $L_i$. 

Now to process an element $A_j$, there are three steps: initialization for intervals that just began; actual processing of $A_j$; and post-processing for intervals that just ended. First, iterate through all $i$ such that $j = a_i$. If $a_i < b_i$, initialize dictionaries $L_{i'}$ for all windows $[i', i'+K)$ in block $i$. If $a_i = b_i$, we cannot afford to store memory, but don't need to: if $O(i) = 1$, output $A_j$ as the answer for each window in block $i$ (otherwise do nothing). Second, to process $A_j$, insert it to all active dictionaries corresponding to windows that contain index $j$. Third, iterate through all $i$ such that $j = b_i > a_i$. If $O(i) = 1$, for each window in block $i$ output $L_i'(d+1)$, where $$d = |\text{Conf}_{l-1, a_i-1}([iB, iB+K))| - |\text{Conf}_{l-1, b_i}([iB, iB+K))|$$ can be shown to be the number of elements of rank less than $l$ caught in $[a_i, b_i]$. Finally, if $j = b_i > a_i$ then deallocate all block $i$ dictionaries. See Figure~\ref{figure:lthalgo} for pseudocode.

\begin{figure}
\fbox{\begin{minipage}{\textwidth}

\begin{algorithm}[\textsc{Process}(1,j)]
{\parindent0pt
\textbf{If} $j = 0$, initialize $N/B$ dictionaries, labelled $L_{iB}$ for $0 \leq i < N/B$.

\textbf{For} $0 \leq i < N/B$,

\qquad \textbf{If} $j \in [iB, iB+K)$, insert $(j,a_j)$ to dictionary $L_{iB}$.
}
\end{algorithm}

\begin{algorithm}[\textsc{Process}(p,j)]
{\parindent0pt
\textbf{If} $j = 0$,

\qquad \textbf{For} $0 \leq i < N/B$,

\qquad \qquad Initialize variables $a_i = b_i = L_{iB}(l)$.

\qquad \qquad \textbf{While} $g_{a_i-1}(iB) < g_{a_i-1}((i+1)B)$, decrement $a_i$.

\qquad \qquad \textbf{While} $g_{b_i}(iB) < g_{b_i}((i+1)B)$, increment $b_i$.

\textbf{For} $0 \leq i \leq N/B$,

\qquad \textbf{If} $j = a_i = b_i$ and $b_i$ is the $(p-1)^\text{th}$ smallest index in $L_i$,

\qquad \qquad Output $A_j$ as the answer for each of the $B$ windows in block $i$.

\qquad \textbf{If} $j = a_i < b_i$,

\qquad \qquad Initialize $B$ dictionaries, labelled $L_{i'}$ for $iB \leq i' < (i+1)B$.

\textbf{For each} active dictionary $L_{i'}$,

\qquad \textbf{If} $j \in [i',i'+K)$, insert $(j,a_j)$ to dictionary $L_{i'}$.

\textbf{For} $0 \leq i \leq N/B$,

\qquad \textbf{If} $j = b_i > a_i$ and $b_i$ is the $(p-1)^\text{th}$ smallest index in $L_i$,

\qquad \qquad Set $d = |\text{Conf}_{l-1,a_i-1}([iB,iB+K))| - |\text{Conf}_{l-1,b_i}([iB,iB+K))|$.

\qquad \qquad \textbf{For} $iB \leq i' < (i+1)B$,

\qquad \qquad \qquad Output $\text{value}(L_{i'}(d+1))$ as the answer for window $[i',i'+K)$.

\qquad \textbf{If} $j = b_i > a_i$, deallocate $L_{i'}$.
}
\end{algorithm}

\end{minipage}}
\caption{The algorithm for \ksmin{}}
\label{figure:lthalgo}
\end{figure}

It's clear that the first pass uses $lN/B$ space and computes the desired values. Taking this for granted, we'll show that in the subsequent $l$ passes, the correct value is computed for every window in every block (ignoring whether it is output, and in what order). Then we'll show that for every window, the output is indeed printed in exactly one of the $l$ passes, and that in any fixed pass, the list of window indices for which output is produced is strictly increasing. Then we will show that the space complexity of each of the last $l$ passes is $\tilde{O}(l^2B)$.

\begin{claim}
Let $0 \leq i < N/B$. For any $iB \leq i' < (i+1)B$, the set $S_{l-1}([i',i'+K)) \cap [a_i, b_i]$ has size $$|\text{Conf}_{l-1,a_i-1}([iB,iB+K))| - |\text{Conf}_{l-1,b_i}([iB,iB+K))|.$$
\end{claim}

\begin{proof}
The number of elements of $S_{l-1}([i',i'+K))$ which exceed $a_i-1$ is $|\text{Conf}_{l-1,a_i-1}([i',i'+K))|$. But by construction, $g_{a_i-1}(iB) = g_{a_i-1}((i+1)B)$. Pairing this equality with the monotonicity shown in Lemma~\ref{gmonotonicity} yields $g_{a_i-1}(iB) = g_{a_i-1}(i')$. The two summands of $g_{a_i-1}$ are similarly monotonic, so in fact $$|\text{Conf}_{l-1,a_i-1}([iB,iB+K))| = |\text{Conf}_{l-1,a_i-1}([i',i'+K))|.$$ A symmetric argument shows that the number of elements of $S_{l-1}([i',i'+K))$ which exceed $b_i$ is $|\text{Conf}_{l-1,b_i}([iB,iB+K))|$.
\end{proof}

\begin{claim}
Let $0 \leq i < N/B$. For any $iB \leq i' < (i+1)B$, we have $s_l([i',i'+K)) \in [a_i, b_i].$
\end{claim}

\begin{proof}
Suppose that $s_l([i', i'+K)) = j$. Then for every $j'$ between $s_l([iB,iB+K))$ and $j$, the $l^\text{th}$-smallest element crosses the boundary $j'/j'+1$ at some point during block $i$. But by Lemma~\ref{gmonotonicity}, $g_{j'}$ increases at each crossing, and is non-decreasing otherwise, so we must have $g_{j'}(iB) < g_{j'}((i+1)B)$. It follows, by construction of $[a_i,b_i]$, that $j \in [a_i, b_i]$.
\end{proof}

Now we can prove correctness of the computations, aside from the issue of output.

\begin{lemma}
Let $0 \leq i < N/B$, and let $iB \leq i' < (i+1)B$. After processing element $b_i$, we have $\text{index}(L_{i'}(d+1)) = s_l([i',i'+K))$.
\end{lemma}

\begin{proof}
Note that after processing element $b_i$, dictionary $L_{i'}$ stores the $l$ smallest elements of $[i',i'+K) \cap [a_i, b_i]$. Of the $l-1$ elements in $[i',i'+K)$ which are smaller than $s_l([i',i'+K))$, exactly $d$ of them (where $d$ is as computed in the algorithm) are in $[i',i'+K) \cap [a_i, b_i]$, as shown in our first claim. From this and from our second claim, which shows that $s_l([i',i'+K))$ is an element of $[i',i'+K) \cap [a_i, b_i]$, we can conclude that the element of rank $d+1$ in $[i',i'+K) \cap [a_i, b_i]$ has index $s_l([i',i'+K))$.
\end{proof}

The next two lemmas address the division of output among the $l$ input passes, and bound the number of output passes.

\begin{lemma}
The above algorithm produces output for each window exactly once.
\end{lemma}

\begin{proof}
We must show that $b_i \in S_l([(i+1)B, (i+1)B+K))$ for each block $i$. There are two cases:
\begin{enumerate}[label = (\roman*)]
\item($a_i < b_i$) Suppose the contrary. Then we have
$$g_{b_i}(iB) \leq g_{b_i-1}(iB) < g_{b_i-1}((i+1)B) = g_{b_i}((i+1)B),$$
where the first inequality follows from monotonicity of the function $G(m) = g_m(iB)$; the second holds since $b_i-1 \in [a_i, b_i]$; and the equality holds since $b_i \not \in S_l([(i+1)B,(i+1)B+K])$. Therefore $g_{b_i}(iB) < g_{b_i}((i+1)B)$, so $b_i+1 \in [a_i, b_i]$. This is a contradiction.
\item ($a_i = b_i$) We know that $s_l([(i+1)B, (i+1)B+K)) \in [a_i, b_i]$, so in fact $$a_i = b_i = s_l([(i+1)B, (i+1)B+K)).$$ 
\end{enumerate}

Hence, $b_i \in S_l([(i+1)B, (i+1)B+K))$ for all $i$, so for every block, output will be produced in exactly one pass.
\end{proof}

\begin{lemma}
The above algorithm uses $l$ output passes.
\end{lemma}

\begin{proof}
We'll show something slightly stronger: for any fixed pass, the order of output is non-decreasing, so that every input pass uses only one output pass. Let $\{i_1, \dots, i_r\}$ be the blocks output in a fixed pass $p$, where $i_1 < i_2 < \dots < i_r$. Observe that the sorted position of $b_{i_k}$ in $S_l([i_kB, i_kB+K))$ is equal to $p-1$, for $1 \leq k \leq r$. Since the sets $S_l([i,i+K))$ are elementwise increasing, it follows that $b_{i_1} \leq b_{i_2} \leq \dots \leq b_{i_r}$. Since block $i$ is output upon processing the element with index $b_i$, and since ties are broken by block index, we can conclude that block $i_k$ is output before block $i_{k+1}$ for $1 \leq k < r$. Within each block, the minima are output in the correct order, so the overall order for each pass is correct.
\end{proof}

The only remaining detail is space complexity.

\begin{theorem} \label{thm:ksmin}
For any $l$, \ksmin{} can be solved in $l+1$ input passes and $l$ outputs passes with $O(l^{3/2} N^{1/2} (\log N + \log R))$ space.
\end{theorem}

\begin{proof}
Correctness has been shown in the above lemmas. It's also been shown that the algorithm uses the desired number of output passes. As for space complexity, the first pass uses $O(N/B)$ dictionaries, and thus $\tilde{O}(lN/B)$ space. For any index $j$, the space used while processing element $j$ in any of the last $l$ passes is $\tilde{O}(lN/B + c_j lB)$, where $c_j$ is the number of blocks $i$ such that $j \in [a_i, b_i]$ and $a_i < b_i$: no space is used for blocks with $a_i = b_i$. But if $j \in [a_i, b_i]$ and $a_i < b_i$, then either $g_{j}(iB) < g_{j}((i+1)B)$ or $g_{j-1}(iB) < g_{j-1}((i+1)B)$. Since $g_{j}$ is a non-decreasing integer function bounded by $0$ and $2l-1$, the first inequality can only hold for $2l-1$ indices $i$. By the same logic, the second inequality can only hold for $2l-1$ values of $i$. Therefore in total, $c_j \leq 4l-2$. We conclude that the space complexity is $\tilde{O}(lN/B + l^2B)$. Setting $B = \sqrt{N/l}$ yields space complexity of $\tilde{O}(l^{3/2}N^{1/2})$.

\end{proof}

\subsection{A Two-Party Communication Algorithm}\label{sec:lcc}

In the communication complexity model, we only care about whether the rank-$l$ element in a window is located in the first half or in the second half of the input. Therefore to simplify notation, we will drop the subscripts ``$b$'', letting $\text{Conf}_l([i,j])$ refer to $\text{Conf}_{l,N/2}([i,j])$, letting $S_\text{cross}$ refer to $S_{N/2-\text{cross}}$, and letting $g(i)$ refer to $g_{N/2}(i)$.

The high-level idea of the communication algorithm for \ksmin{} below is to compute the set of indices $S_\text{cross} = \{S_1,\dots,S_{|S_\text{cross}|}\}$ where the $l^\text{th}$-smallest element crosses from one party's input to the other's. Pick any $c$ with $1 \leq c \leq |S_\text{cross}|$. By Lemma~\ref{gmonotonicity}, $g(i) < g(S_c)$ for all $i < S_c$ and $g(i) \geq g(S_c)$ for all $i \geq S_c$. Therefore a simple $d$-ary search algorithm to find all (at most $2l-1$) indices at which $g$ increases, will produce a set of candidate indices which contains $S_\text{cross}$. For this small set of candidates, it is inexpensive to check which indices are in fact crossing indices. Now the parties can alternate producing output, handing control to the other party when a crossing index is reached.

Here is a more detailed description, with $D = N^{1/(p-2l-1)}$ the search width parameter.

\begin{algorithm}
{\parindent0pt
1. For all $v \in [2l-1]$ in parallel, perform the following search algorithm for candidate $C_v$.

\qquad Initialize the search interval as $[\text{low},\text{high}) = [N/2-K-1,N/2]$.

\qquad \textbf{While} $\text{high} - \text{low} > 1$,

\qquad \qquad Alice picks $D$ evenly spaced indices $\{i_1,\dots,i_D\}$ in $[\text{low},\text{high}]$.

\qquad \qquad For each index $i_j$, Alice communicates $i_j$ and $\{A_i \mid i \in S_l([i_j,N/2))\}$ to Bob.

\qquad \qquad For each index $i_j$, Bob computes $\text{Conf}_{l-1}([i_j,i_j+K))$, $\text{Conf}_l([i_j,i_j+K))$, and then $g(i_j)$.

\qquad \qquad Bob sets $[\text{low},\text{high}) = (i_{j^*},i_{j^*+1}]$ where $j^* = \max \{j \mid g(i_j) < v\}$. 

\qquad \qquad If loop not over, Bob and Alice switch roles (so next time Bob picks $D$ indices)

2. Assume the loop ends with Bob's turn; if not, an extra round is used, in which Alice communicates all $C_v$ to Bob. For each $v \in [2l-1]$, Bob communicates $C_v$ and $\{A_i \mid i \in S_l([N/2,C_v+K))\}$ and $\{A_i \mid i \in S_l([N/2,C_v+K-1))\}$ to Alice. With this information, Alice computes $\text{Conf}_l([C_v,C_v+K))$ and $\text{Conf}_l([C_v-1,C_v+K-1))$, and determines whether $C_v \in S_\text{cross}$. She also determines $|\text{Conf}_l([i,i+K))|$ for all indices $i$.

3. For each $i \in [0,S_1)$, Alice outputs the rank-$l'(i)$ element of window $[i, \min(i+K,N/2))$, where $l'(i) = l - |\text{Conf}_l([N/2,i+K))|$. Then Alice communicates $S_\text{cross}$ to Bob, and all candidate indices at which the size of the $l$-configuration increases.

4. Bob determines $|\text{Conf}_l([i,i+K))|$ for all indices $i$. Then for each $i \in [S_1, S_2)$, Bob outputs the rank $l''(i)$ element of window $[N/2,i+K)$, where $l''(i) = |\text{Conf}_l([N/2,i+K))|$.

5. Alice and Bob alternate producing output until done.
}
\end{algorithm}

\begin{claim}
The above algorithm is correct.
\end{claim}

\begin{proof}
Note that it is possible for Bob to compute an $l$-configuration given his own input and the $l$ smallest elements of the intersection of the window with Alice's input. And $g$ can be computed from an $l$-configuration and an $(l-1)$ configuration. Hence the search algorithm is feasible.

If $i \in S_\text{cross}$, then $g(i) > g(i-1)$ by Lemma~\ref{gmonotonicity}, so $i$ is the first index at which $g$ reaches $g(i)$. Hence $i = C_{g(i)}$. Therefore the set of candidates contains $S_\text{cross}$.

Computing the configurations in step $(2)$ is feasible, as already argued. From these one can check whether the candidate is a crossing index, by checking if $l$ enters or exits the configuration. And Alice can compute the size of every $l$-configuration, since any location at which the size increases is a candidate index $C_v$ for some $v$, and Alice can determine whether the size increases at $C_v$ since she knows the entire $l$-configurations at $[C_v-1,C_v+K-1)$ and $[C_v,C_v+K)$.

Bob can compute $S_\text{cross}$ and the size of every $l$-configuration, from the information sent by Alice. Finally, it's clear that if the rank-$l$ element for any window $[i,i+K)$ lies in Alice's input, its rank in $[i,\min(i+K,N/2))$ is exactly $l - |\text{Conf}_l([N/2,i+K))|$, and similarly for Bob.
\end{proof}

\begin{theorem}\label{ccgen}
Let $l \geq 1$ and $p \geq 2l+2$. In the $p$-round communication complexity model, \ksmin{} can be solved using $\tilde{O}(pl^2 N^{1/(p-2l-1)})$ bits of communication.
\end{theorem}

\begin{proof}
Step $(1)$ of the algorithm takes $p - 2l - 1$ rounds. Each $v$ contributes $O(l D \log R)$ bits to each round, yielding total communication $O(l^2 D \log R)$ per round. Step $(2)$ uses at most $2$ rounds, with $O(l \log N + l^2 \log R)$ communication. The remaining steps use $2l$ turns and thus $2l-1$ rounds, since $|S_\text{cross}| \leq 2l-1$. In the first round after Alice has output some answers, she sends $O(l \log N)$ bits to Bob, and no communication need occur in the last $2l-2$ rounds.
\end{proof}

Note that in the current statement, the bound in Theorem~\ref{ccgen} does not match the bound in Proposition~\ref{cc3} when $l = 1$. A constant number of rounds can be eliminated from the protocol by methods such as that employed in Proposition~\ref{cc3}, but such optimizations were not the focus of the proof, so we have ignored them for now.

\section{A Lower Bound for \smin{}}\label{sec:lowerbounds}

One might hypothesize that increasing the number of passes from two to three or more would allow even more space-efficient algorithms for \smin. However, it turns out that for large $R$, the above algorithm is asymptotically optimal up to a logarithmic factor---at least, among streaming algorithms which use any constant number of input passes, and only one output pass. To prove this, we start by lower bounding the difficulty of \smin{} in the one-round communication complexity model.

Fix an input size $N$, a window size $K$, and a maximum element value $R$. Assume that $N \geq 2K$ and $R \geq K$. We'll provide a lower bound for the case $N = 2K$, since increasing $N$ further can only increase the difficulty of the problem. Our goal is to construct a set of inputs $S$ such that any communication algorithm which uses few bits of communication can only solve a small fraction of the input. In particular, let the set of Alice's inputs be $$A = \{a_1, \dots, a_K \mid 0 < a_1 \leq \dots \leq a_K < R\}.$$ For $0 \leq i \leq K/(2M)$, let $B(i)$ be the sequence of $2Mi$ $R$'s, followed by $K-2Mi$ $0$'s. Then let the set of Bob's inputs be $B = \{B(0),\dots,B(K/(2M))$. Let $S = A \times B$, so that every input consists of an input for Alice picked from $A$, concatenated with an input for Bob picked from $B$. For notational convenience, we will use $a$ to denote an element of $A$ and use $b$ to denote an element of $B$.

This choice of inputs is motivated by the following observation. As the window slides through the input, the index of the minimum element never decreases, and eventually the minimum element enters Bob's input. If Alice knew the index of the last window for which the minimum was in her input, she could determine the minimums of all windows up until that last window, and Bob could determine the rest. Conversely, intuition suggests that if Alice does not know that index, it is hard for Alice and Bob to determine all minimums in one round. And we have chosen a set of inputs $S$ for which knowing Alice's input does not yield any information about the index of the last window. Therefore intuitively, for most inputs either Alice produces too much output or Bob produces too much output, so we can bound the fraction of inputs for which the output is correct. The following lemma works out the exact bounds.

\begin{lemma}\label{lem:badinputs}
Consider the two-party communication complexity model of \smin{} in which the window size is $K$, the input size is $N \geq 2K$, and the maximum element value is $R \geq K$. Let $M \geq 2 \log K$. Then there is a fixed set of inputs $S$ and a constant $c$ so that no $1$-round deterministic algorithm using $M$ bits of communication solves \smin{} (that is, outputs the correct $N+1-K$ sliding window minimums) on more than $cM|S|/K$ inputs in $S$.
\end{lemma}

\begin{proof}
For any input $(a,b) \in S$, Alice outputs several integers, which are supposed to equal the minimum values in the first several windows of size $K$. Then Alice passes a message of length $M$ to Bob, and Bob outputs several integers. Let $\text{out}_A(a)$ be the number of integers output by Alice; crucially, this number does not depend on Bob's input $b$. However, the number of minimums in Alice's input does. Suppose Bob's input is $b = B(i)$. Then the minimums of the first $2Mi$ windows are in Alice's input, and the minimums of the remaining $K - 2Mi$ windows are in Bob's input, and are equal to $0$.

There are two broad cases we'll bound separately. Either Alice produces ``too much'' output, in which case we obtain a straightforward bound on the fraction of such inputs which she solves; or Bob produces ``too much'' output. In the latter case, we subdivide by the exact value of $\text{out}_A(a)$ and bound the fraction of such inputs solved by Bob, in terms of the communication used.
\begin{enumerate}
\item We bound the number of solved inputs $(a,B(i)) \in S$ such that $\text{out}_A(a) > 2Mi - 2M$. Fix $a \in A$. The number of positive window minimums among the first $\text{out}_A(a)$ windows of input $(a,B(i))$ is different for each $i \in [0,K/(2M)]$ satisfying the above inequality. But the number of positive integers output by Alice is completely determined by her input $a$, so there is at most one $i \in [0, K/(2M)]$ satisfying $\text{out}_A(a) > 2Mi - 2M$ for which Alice gives the correct output. In total, there are at most $2M|S|/K$ inputs satisfying the inequality, for which correct output is given.

\item We bound the number of solved inputs $(a,B(i)) \in S$ such that $\text{out}_A(a) \leq 2Mi - 2M$. Pick some $v \in [0,K]$ and let $A_v = \{ a \in A \mid \text{out}_A(a) = v\}$. For any $a \in A_v$ and $B(i) \in B$, if $v = \text{out}_A(a) \leq 2Mi - 2M$, then Bob must output $a_{2Mi-2M+1}, a_{2Mi-2M+2}, \dots, a_{2Mi}$, since these $2M$ integers in Alice's input are the minimums of $2M$ windows which Alice should have output but did not. Fix some $B(i) \in B$ with $v \leq 2Mi - 2M$. Then Bob can output at most $2^M$ different ``guesses" for $a_{2Mi-2M+1}, \dots, a_{2Mi}$ on different inputs $(a,B(i))$, one for each communication string of length $M$ sent by Alice. So if Alice's input $a \in A_v$ has an assignment of $a_{2Mi-2M+1}, \dots, a_{2Mi}$ which is not in Bob's $2^M$ guesses, then Bob will produce incorrect output.

We want to bound the number of inputs $a \in A_v$ with assignments of the elements $a_{2Mi-2M+1}, \dots, a_{2Mi}$ that do match one of Bob's guesses. Fix one of Bob's guesses for these $2M$ elements. The number of non-decreasing sequences of length $K$ in $[1,R-1]$ which match this guess, is at most the number of non-decreasing sequences of length $K-2M$ in $[1, R-1]$, which is approximately $\binom{K-2M+R-1}{K-2M}$. Summing over all of Bob's $2^M$ guesses, there are at most $2^M \binom{K-2M+R-1}{K-2M}$ inputs $a \in A_v$ so that input $(a, B(i))$ yields correct output.

Summing over all $v \in [0,K]$ and $i \in [0, K/(2M)]$, the number of inputs $(a,B(i))$ with $\text{out}_A(a) \leq 2Mi - 2M$ for which correct output is produced is at most $(K^2/M) 2^M \binom{K-2M+R-1}{K-2M}$ . But $|A| = \binom{R+K-1}{K}$ and $|B| = K/M$, so the total number of inputs under consideration is $|S| = (K/M) \binom{R+K-1}{K}$. Repeatedly applying the identity $\binom{n}{k} = \frac{n}{k} \binom{n-1}{k-1}$ to the binomial coefficient and using the assumption $R \geq K$, we can show that the fraction of inputs on which correct output is produced is at most $K/2^{M}$. Under the assumption $M \geq 2 \log K$, this fraction is at most $1/K$.

\end{enumerate}

Summing over both cases the fractions of input which yield correct output, and setting the constant $c$ appropriately, yields the desired result.
\end{proof}

It remains to extend the above communication bound to a streaming bound. We want a stronger streaming bound than is possible to obtain in the communication model, so we must ``compound'' the communication bound somehow. Fix an input size $N$ and a space bound $M = O(\sqrt{N})$. Suppose an algorithm \textsc{algo} in the $p$-pass streaming model uses $M$ space to solve \smin{} for any $K$, when $R \geq K$. Our bound will be a worst-case bound over all window sizes $K$; so we'll choose $K = 2cM$. For these parameters, define a set of inputs $S'$ by $S' = S^{N/(2K)}$, so that for each block of size $2K$ an input is independently chosen from the set $S$ described in Lemma~\ref{lem:badinputs}.

By the following claim, we can identify a large set of inputs $S'' \subseteq S'$ for which a large fraction of the output is produced in a single, fixed pass.

\begin{claim}
There is a pass index $p'$, an input index $i'$, and a set $S'' \subseteq S'$ of size at least $|S'|/(pN)$, such that \textsc{algo} outputs the window minimums with starting indices in $[i', i' + (N-K)/p]$ during pass $p'$, for every input $s \in S''$.
\end{claim}

\begin{proof}
For any input, there is some pass during which the algorithm outputs at least $(N+1-K)/p$ window minimums, by the pigeonhole principle. By our running assumption (unless stated otherwise) that only a single output pass is used, the window minimums output in one pass are the minimums of consecutive windows. Therefore applying pigeonhole again, there is some pass index $p'$ and some input index $i'$ so that for at least $|S'|/(pN)$ inputs in $S'$, the algorithm outputs the window minimums with starting indices $\{i', i'+1, \dots, i' + (N-K)/p\}$ during pass $p'$. Let the set of inputs in $S'$ satisfying this condition be denoted by $S''$.
\end{proof}

We can now focus entirely on pass $p'$, reducing the multi-pass algorithm to a one-pass algorithm with $M$ bits of advice. There are $2^M$ possibilities for the memory stored by \textsc{algo} immediately before pass $p'$. Fixing any one such possibility induces a one-pass deterministic algorithm. There must be at least one possibility which occurs for at least $|S''|/2^M$ inputs in $S''$. Call this algorithm \textsc{solve-p$'$}, and let $T \subseteq S''$ be the set of inputs which induce \textsc{solve-p$'$}. To bound $M$, we'll provide an opposing upper bound on $|T|/|S'|$ and therefore on $|T|/|S''|$. To bound $|T|/|S'|$, we'll show that for many blocks of size $2K$, the set $T$ only takes on a fraction of the potential values on that block. The following claim is the key.

\begin{claim}
There are $N/(4pK)$ blocks such that for all inputs in $T$, \textsc{solve-p$'$} outputs the (correct) window minimums for all windows with starting indices in any of these blocks.
\end{claim}

\begin{proof}
By the properties of \textsc{algo} on $S''$, and since $T \subseteq S''$ is the set of inputs for which \textsc{algo} induces \textsc{solve-p$'$} for pass $p'$, it's clear that for any input in $T$, \textsc{solve-p$'$} outputs the correct minimums for all windows with starting indices in $[i', i' + (N-K)/p]$. This interval contains at least $((N-K)/(2pK)) - 2$ blocks. Since $K = O(\sqrt{N})$, we have $N \geq (8p+2)K$ for sufficiently large $N$, so the number of blocks is at least $N/(4pK)$.
\end{proof}

For each block, we'd like to use \textsc{solve-p$'$} to simulate a one-round communication algorithm solving all inputs $s \in S$ achieved by $T$ on that block, so that Lemma~\ref{lem:badinputs} applies. However, this would require the following procedure: for any $(a,b) \in S$ which is achieved by $T$ on a fixed block, find some string $s$ such that $sab$ is a prefix of a string in $T$, without knowing $b$. In general, this cannot be done. However, it can be done for a constant fraction of inputs, after discarding a constant fraction of the blocks. 

Formally, for $1 \leq i \leq N/(2K)$ and $s \in (A \times B)^{i-1}$ and $a \in A$, let $m_{i,s,a}$ be the number of $b \in B$ such that there exists some $s' \in (A \times B)^{N/(2K)-i}$ with $sabs' \in T$. Also, for $1 \leq i \leq N/(2K)$, let $n_i$ be the number of $a \in A$ such that there exists some $s \in (A \times B)^{i-1}$ for which $m_{i,s,a} \geq 9|B|/10$. If $n_i/|A|$ is bounded above zero, then we can achieve a nontrivial bound for block $i$, as seen in the following simulation lemma.

\begin{lemma}
Let $[2Kj, 2K(j+1))$ be a block such that $n_j \geq 9|A|/10$. Then the number of $(a,b) \in S$ such that some element of $T$ contains $ab$ at block $j$ is at most $3|S|/5$.
\end{lemma}

\begin{proof}
We want to bound the number of $a \in A$ and $b \in B$ such that some element of $T$ contains $ab$ at block $j$: formally, $sabs' \in T$ for some prefix $s \in (A \times B)^{j}$ and suffix $s' \in (A \times B)^{N/(2K) - j - 1}$. There are two cases to consider. In the first case, let $a \in A$ be such that for every $s \in (A \times B)^j$ we have $m_{i,s,a} \leq 9|B|/10$. There are only $|A| - n_j \leq |A|/10$ such choices of $a$. For each choice of $a$, there are at most $|B|$ choices of $b$; so in total there are at most $|S|/10$ choices of $a$ and $b$ satisfying the given property.

In the second case, we want to bound the number of $a \in A$ and $b \in B$ such that some element of $T$ contains $ab$ at block $j$, and $\max_{s \in (A \times B)^j} m_{i,s,a} \geq 9|B|/10$. Consider the following one-round communication algorithm, in which Alice is given an input $a \in A$ and Bob is given an input $b \in B$. Alice picks some $s = s_a \in (A \times B)^j$ which maximizes $m_{i,s,a}$ (the maximal value may be $0$). She runs \textsc{solve-p$'$} on $sa$, feeding to her output stream all output for windows in $ab$. She communicates the final memory state of \textsc{solve-p$'$} to Bob, who initializes \textsc{solve-p$'$} to that state and runs the algorithm on $b$. Consider the state $Q$ of the algorithm after processing $sab$. When \textsc{solve-p$'$} is run on elements of $T$, it eventually produces correct output for all windows in block $j$. So there is a function from states $Q$ to remaining output sequences, which Bob can apply to print the remaining output. If $sab$ is a prefix of some element of $T$, then Alice's output is correct, as is Bob's output. Let $I$ be the set of pairs $(a,b)$ such that some element of $T$ contains $ab$ at block $j$, and $\max_{s \in (A \times B)^j} m_{i,s,a} \geq 9|B|/10$; let $I_\text{correct} \subseteq I$ be the set of such pairs for which $s_aab$ is a prefix of some element of $T$. If $(a,b) \in I \setminus I_\text{correct}$, then there are only $|B|/10$ choices for $b$, so $|I \setminus I_\text{correct}| \leq |A| \cdot |B|/10 = |S|/10$. But we also have $|I_\text{correct}| \leq cM|S|/K = |S|/2$ by Lemma~\ref{lem:badinputs}, so $|I| \leq 11|S|/20$.

In total, the number of $a\in A$ and $b \in B$ such that some element of $T$ contains $ab$ at block $j$ is at most $|S|/10 + 11|S|/20 = 3|S|/5$.
\end{proof}

It only remains to show that $n_i \geq 9|A|/10$ for sufficiently many blocks. This follows from the density of $T$ in $S'$ and the following, purely combinatorial, lemma.

\begin{lemma}\label{lem:sparse}
Let $A$ and $B$ be sets, and let $T \subseteq (A \times B)^d$. For $1 \leq i \leq d$ and $s \in (A \times B)^{i-1}$ and $a \in A$, let $m_{i,s,a}$ be the number of $b \in B$ such that there exists some $s' \in (A \times B)^{d-i}$ with $sabs' \in T$. Also, for $1 \leq i \leq d$, let $n_i$ be the number of $a \in A$ such that there exists some $s \in (A \times B)^{i-1}$ for which $m_{i,s,a} \geq 9|B|/10$. Then the set of $i \in [d]$ for which $n_i \leq 9|A|/10$ has size at most $c_{\text{log}} \log \frac{(|A| \cdot |B|)^d}{|T|}$, where $c_{\text{log}} = 1/\log \frac{100}{99}$. 
\end{lemma}

\begin{proof}
For $0 \leq i \leq d$, define $P_i$ to be the set of $s \in (A \times B)^i$ for which there exists some $s' \in (A \times B)^{d-i}$ with $ss' \in T$.

Let $i \in [d]$ and suppose $n_i \geq 9|A|/10$. Then we simply apply the trivial bound $|P_i| \leq |A| \cdot |B| \cdot |P_{i-1}|$. Now suppose that $n_i \leq 9|A|/10$. Let $A_1$ be the set of $a \in A$ such that $m_{i,s,a} \leq 9|B|/10$ for every $s \in (A \times B)^{i-1}$. Fix some $p \in P_{i-1}$. Let $a \in A$. If $a \in A_1$, then $m_{i,p,a} \leq 9|B|/10$. So there are at most $9|B|/10$ choices of $b \in B$ such that $pab \in P_i$. And if $a \not \in A_1$, then there are at most $|B|$ choices. In total, the number of choices of $a \in A$ and $b \in B$ such that $pab \in P_i$ is at most \begin{align*} |A_1| \cdot \frac{9|B|}{10} + (|A| - |A_1|) \cdot |B| &= |A| \cdot |B| - |A_1| \cdot \frac{|B|}{10} \\&\leq \frac{99 |A| \cdot |B|}{100},\end{align*} where the last inequality holds since $|A_1| \geq |A|/10$. But every element of $P_i$ has the form $pab$ for some $p \in P_{i-1}$, $a \in A$, and $b \in B$, so we get the inequality $|P_i| \leq 99 |P_{i-1}| \cdot |A| \cdot |B| / 100.$

Since we have $|P_0| = 1$ and $|P_d| = |T|$, it follows that $$|T| \leq (|A| \cdot |B|)^d \cdot \left ( \frac{99}{100} \right)^k,$$ where $k$ is the number of $i \in [d]$ for which $n_i \leq 9|A|/10$. Taking logarithms yields the desired bound.
\end{proof}

Specifically, we get the following bound on the number of ``useful'' blocks.

\begin{claim}
If $M \leq (32pcc_\text{log})^{-1/2}N^{1/2}$, where $c$ and $c_\text{log}$ are the constants determined in above lemmas, then the set of $i \in [N/(2K)]$ for which $n_i \leq 9|A|/10$ has size at most $N/(8pK)$.
\end{claim}

\begin{proof}
By Lemma~\ref{lem:sparse}, we have that the set of $i \in [N/(2K)]$ for which $n_i \leq 9|A|/10$ has size at most $$c_\text{log} \log \frac{|S'|}{|T|} \leq c_\text{log} \log (pN \cdot 2^M) = c_\text{log}(M + \log pN) \leq N/(16pcM) = N/(8pK)$$ as desired.
\end{proof}

Thus, there are at least $N/(8pk)$ blocks $[2Kj, 2K(j+1))$ for which the number of $a\in A$ and $b \in B$ such that some element of $T$ contains $ab$ at block $j$ is at most $3|S|/5$. Identifying every input in $S''$ with a tuple of $N/(8pK)$ block inputs in $S$, the set $T$ cannot contain more than $(3|S|/5)^{N/(8pK)}$ distinct tuples.

Using this, we can bound the size of $T$. A set of inputs in $S'$ which all take the same values in the $N/(8pK)$ relevant blocks can have size at most $|S'|/|S|^{N/(8pK)}$. Hence a set of inputs in $S'$ which take at most $(3|S|/5)^{N/(8pK)}$ values in the $N/(8pK)$ relevant blocks can have size at most $|S'|(3/5)^{N/(8pK)}$. So $T$ can have size at most $|S'|(3/5)^{N/(8pK)}$ as well. By the bound $|S''| \geq |S'|/(pN)$, we conclude that \textsc{solve-p$'$} solves at most $pN|S''|(cM/K)^{N/(4pK)}$ inputs in $S''$.

But we assumed that $|T| \geq |S''|/2^M$, so \textsc{solve-p$'$} solves at least $|S''|/2^M$ inputs in $S''$. Thus $2^M \geq (pN)^{-1} (5/3)^{N/(8pK)}$. Substituting $K = 2cM$ and solving for $M$ yields $$M \geq -\frac{1}{2}\log pN + \frac{1}{4pc}\sqrt{4p^2c^2\log^2 pN + pcN\log \frac{5}{3}},$$ which is sufficient to prove the following theorem.

\begin{theorem}\label{thm:lbound}
For any constant $p$, no algorithm in the $p$-pass streaming model solves \smin{} in $o(N^{1/2})$ space for all window sizes $K$ when the maximum element value is $R = K$. 
\end{theorem}




\section{A Linear Lower Bound for \smaj{}}\label{sec:maj}

Now we show that for any constant number of passes, \smaj{} has no algorithm with space complexity sublinear in $N$ and independent of $K$. The proof broadly follows the same structure. We start with a lemma lower bounding the difficulty of \smaj{} in the two-party communication model. We then consider a streaming algorithm which solves \smaj{}, and pick a single pass in which the algorithm often produces significant input. However, now we need $K = \Theta(N)$ in order to obtain a linear lower bound on $M$, so we could only have a constant number of ``blocks''. Therefore the above method of discarding up to $c_\text{log}(M + \log pN)$ sparse blocks does not work. Instead, since we only need a single block, we only have two cases, which we'll bound separately: either the algorithm starts producing output for the block early, or it doesn't. The former case lends itself to an information bound, and the latter allows simulation of a communication algorithm.

Let's return to bounding the difficult of \smaj{} in the communication model. Pick any $K$. Let $A$ be the set of boolean strings $a$ such that $|a| = K-1$ and $a_{2i} + a_{2i+1} = 1$ for $0 \leq i < (K-1)/2$. Let $B$ be defined in the same way, and define $S = A \times B$, so that every input consists of a string $a \in A$ for Alice, and a string $b \in B$ for Bob.

For each $a \in A$, define $\text{out}_A(a)$ to be the number of bits output by Alice. So $0 \leq \text{out}_A(a) \leq K-1$ for all $a \in A$. Observe that if Alice outputs $v$ bits of (correct) output, then her output fixes the first $2 \lfloor \frac{v+1}{2} \rfloor$ bits of $b$. For example, if Alice's first bit of output is correctly $0$, then $b_0 = 0$ and $b_1 = 1$, whereas if her first bit of output is correctly $1$, then $b_0 = 1$ and $b_1 = 0$. Alice's third bit of output similarly fixes $b_2$ and $b_3$. The following lemma will bound the number of inputs in $S$ which any low-memory communication algorithm can solve, by caseworking on $\text{out}_A(a)$.

\begin{lemma}\label{lem:majbadinputs}
Consider the two-party communication complexity model of \smaj{} in which the window size is $K$, the input size is $N \geq 2K - 2$, and each party receives a contiguous half of the input. Let $M \geq 2 \log K$. Then there is a fixed set of inputs $S = S_N$ and constants $c_1, c_2$ so that no $1$-round deterministic algorithm using $M$ space solves \smaj{} (that is, outputs the correct $N+1-K$ sliding window majority bits) on more than $2^{c_1M - c_2K}|S|$ inputs in $S$.
\end{lemma}

\begin{proof}
Let $S^* \subseteq S$ be the set of inputs for which Alice and Bob produce the correct output. For $0 \leq v \leq K-1$, let $A_v = \{a \in A \mid \text{out}_A(a) = v\}$. We'll bound the size of each $(A_v \times B) \cap S^*$ separately.

Fix some $v$ with $0 \leq v \leq K-1$. Pick any $a \in A_v$. For any $(a,b) \in A_v \cap S^*$, Alice's $v$ bits of output fix the first $2 \lfloor \frac{v+1}{2} \rfloor$ bits of $b$. Hence, $$|(A_v \times B) \cap S^*| \leq \frac{|A_v| \cdot |B|}{2^{\lfloor \frac{v+1}{2} \rfloor}}.$$

The above bound is strong when $v$ is large. When $v$ is small, we use Bob's output. Fix some $b \in B$. There are $2^M$ possible messages Bob can receive from Alice; fix one of them. Then Bob's output is fixed, and it fixes the last $2\lfloor \frac{K-v}{2} \rfloor$ bits of any $a \in A_v$ for which $(a,b) \in S^*$. By the same logic as above, but this time summing over all $2^M$ messages which Bob could receive, $$|(A_v \times B) \cap S^*| \leq \frac{|A_v| \cdot |B|}{2^{\lceil \frac{K-v}{2} \rceil - M}}.$$

Using the first bound when $v \geq \frac{K-2M-1}{2}$ and the second bound otherwise, we get $$|(A_v \times B) \cap S^*| \leq \frac{|A_v| \cdot |B|}{2^{\frac{K-2M+1}{4}}}$$ for any $0 \leq v \leq K-1$. Summing over $v$, it follows that \[|S^*| \leq 2^{\frac{2M-K-1}{4}} |S|.\qedhere\]
\end{proof}

Now let's use this for a streaming lower bound. Fix an input size $N$, and let $K$ be the largest odd integer below $N/(2p + 1)$. Suppose that an algorithm \textsc{algo} in the $p$-pass streaming model uses $M$ space to solve \smaj{} for these parameters. Let $S$ be the set of boolean strings $s$ of length $N$ such that $s_{2i} + s_{2i+1} = 1$ for $0 \leq i < N/2$. As we saw in the previous section, there is some $S'' \subseteq S$ with $|S''| \geq |S|/(pN)$, so that for every input in $S''$, \textsc{algo} outputs the majority bits for the windows with starting indices $\{i', i'+1, \dots, i' + (N+1-K)/p - 1\}$ during pass $p'$. By choice of $K$, we have $(N+1-K)/p \geq 2K$. Let $S''_\text{early} \subseteq S''$ be the subset of $S''$ consisting of inputs for which window $i'+K-1$ is output before processing element $i'+K-1$. We'll bound $S''_\text{early}$ and $S \setminus S''_\text{early}$ separately.

\begin{lemma}
We have $|S''_\text{early}| \leq 2^{N/2 + M - K/2}$.
\end{lemma}

\begin{proof}
There are at most $2^M$ possibilities for the memory stored by \textsc{algo} immediately before starting pass $p'$. Each possibility induces a one-pass algorithm. We label the algorithms as \textsc{full-$i$} for $i \in [2^M]$, and let $T_\text{early}(i)$ be the set of inputs in $S''$ such that \textsc{full-$i$} outputs the majority bits for the windows with starting indices $\{i',\dots,i'+2K-1\}$, and outputs the bit for window $i'+K-1$ before processing element $i'+K-1$. Then $\cup_{i \in [2^M]} T_\text{early}(i) \supseteq S''_\text{early}$. 

Fix any $i \in [2^M]$. To identify an input in $T_\text{early}(i)$, it suffices to choose the first $i'+K-1$ elements, which fixes \textsc{full-$i$}'s output for windows $\{i',\dots,i'+K-1\}$, which in turn fixes the next $K$ input elements (by the same reasoning as in Lemma~\ref{lem:majbadinputs}); and then to choose the remaining $N + 1 - i' - 2K$ input elements. Every pair of elements requires a single boolean choices. Hence there are $N/2 - K/2$ boolean choices to make, so $|T_\text{early}(i)| \leq 2^{N/2 - K/2}$. Summing over all $i \in [2^M]$, it follows that $|S''_\text{early}| \leq 2^{N/2 + M - K/2}$.
\end{proof}

\begin{lemma}
We have $|S'' \setminus S''_\text{early}| \leq 2^{(c_1+1)M - c_2K + N/2}$, where the constants are as in Lemma~\ref{lem:majbadinputs}.
\end{lemma}

\begin{proof}
Consider the $2^M$ possibilities for the memory stored by \textsc{algo} immediately before processing element $i'+K-1$ in pass $p'$. Each possibility induces a one-pass algorithm, and we label the algorithms as \textsc{partial-$i$} for $i \in [2^M]$. Let $T_\text{punctual}(i)$ be the set of inputs in $S''$ such that \textsc{partial-$i$} outputs the majority bits for the windows with starting indices $\{i'+K-1,\dots,i'+2K-1\}$. For any element $s \in S'' \setminus S''_\text{early}$, running \textsc{algo} over $s$ for $p'-1$ passes, and then over the first $i'+K-1$ elements of $s$ once, yields some memory state $i_s \in [2^M]$. As \textsc{partial-$i_s$} outputs the majority bits for the windows on $s$ with starting indices $\{i'+K-1,\dots,i'+2K-1\}$, it follows that $\cup_{i \in [2^M]} T_\text{punctual}(i) \supseteq S'' \setminus S''_\text{early}$.

Fix any $i \in [2^M]$. Here's a one-round communication algorithm on boolean strings $s$ of length $2K-2$ where $s_{2j} + s_{2j+1} = 1$ for $0 \leq j < K-1$. Alice receives string $a$, consisting of the first $K-1$ bits of input, and Bob receives string $b$, consisting of the remaining $K-1$ bits. Alice runs \textsc{partial-$i$} on $a$, and communicates the final memory state to Bob. Then Bob initializes \textsc{partial-$i$} to the given state, and runs it on $b$. At this point, if \textsc{partial-$i$} has not produces $K-1$ bits of output, the final memory state encodes the remaining bits, which Bob can compute and output. If there is some $t \in T_\text{punctual}(i)$ such that the input $s = ab$ is equal to the subsequence of $t$ between indices $i'+K-1$ and $i'+3K-4$, then the communication algorithm is correct on $s$. But the algorithm uses only $M$ bits of communication, so by Lemma~\ref{lem:majbadinputs}, it is correct on at most $2^{c_1M - c_2K}2^{K-1}$ inputs, for constants $c_1, c_2 > 0$. Each such input can correspond to at most $2^{N/2 - (2K-2)/2}$ elements of $T_\text{punctual}(i)$, so $|T_\text{punctual}(i)| \leq 2^{c_1M - c_2K + N/2}$. Summing over all $i \in [2^M]$, we obtain $|S'' \setminus S''_\text{early}| \leq 2^{(c_1+1)M - c_2K + N/2}$.
\end{proof}

The proof is essentially complete. Applying the above lemmas, \begin{align*}
|S''|/|S| 
&= 2^{-N/2}(|S''_\text{early}| + |S'' - S''_\text{early}|) \\
&\leq 2^{M - K/2} + 2^{(c_1+1)M - c_2K} \\
&\leq 2^{(c_1+1)M - K/2 + 1}.
\end{align*}
By assumption, $|S''| \geq |S|/(pN)$. So $2^{(c_1 + 1)M - K/2 + 1} \geq 1/(pN)$. Solving for $M$ and using $K = \Theta(N)$, the desired result follows:

\begin{theorem} \label{thm:majlbound}
For any constant $p$, no algorithm in the $p$-pass streaming model solves \smaj{} in $o(N)$ space for all windows sizes $K$.
\end{theorem}

\section{Future Work}
Our results leave a number of questions unanswered. Is it possible to decrease the dependence on $l$ in our streaming algorithm for \ksmin{}, or is there a matching lower bound of $\widetilde{O}(l^{3/2} \sqrt{N})$ (at least, for $l \leq N^{1/3}$)? We might ask a similar question in the communication complexity model. And our lower bounds only apply when the algorithm is restricted to a single output pass; can the bound be extended to $l$ output passes, or to any constant?
 
Another issue is the logarithmic gap between our $2$-pass streaming algorithm for \smin{} and our lower bound. It seems difficult to find an algorithm without the logarithmic factor, so we conjecture that the lower bound can be improved to $\Omega(N^{1/2} \log N)$. The case $R = o(K)$ still remains open; there may be an improved algorithm, possibly using more passes, or there may be a matching lower bound.


\begin{thebibliography}{10}
\bibitem{arasu2004}
A. Arasu and G. S. Manku. \textit{Approximate counts and quantiles over sliding windows}, in Proceedings of the twenty-third ACM SIGMOD-SIGACT-SIGART symposium on Principles of database systems (PODS '04). ACM, New York, NY, USA, 286--296, 2004.

\bibitem{babcock2003}
B. Babcock, M. Datar, R. Motwani, and L. O'Callaghan. \textit{Maintaining variance and k-medians over data stream windows}, in Proceedings of the 9th
ACM SIGMOD Symposium on Principles of Database Systems (PODS), pages 234--243, 2003.

\bibitem{braverman2007}
V. Braverman and R. Ostrovsky. \textit{Smooth histograms for sliding windows}, in Proceedings of the 48th IEEE Symposium on Foundations of Computer Science (FOCS), pages 283--293, 2007.

\bibitem{braverman2010}
V. Braverman and R. Ostrovsky. \textit{Effective computations on sliding windows}, in SIAM Journal on Computing, 39(6):2113--2131, 2010.

\bibitem{braverman2012}
V. Braverman, R. Ostrovsky, and C. Zaniolo. \textit{Optimal sampling from sliding windows}, in Journal of Computer and System Sciences, 78(1):260--272, 2012.

\bibitem{braverman2013}
V. Braverman, R. Gelles, and R. Ostrovsky. \textit{How to catch L2-heavy-hitters on sliding windows}, in Proceedings of the 19th Computing and Combinatorics Conference, pages 638--650, 2013.

\bibitem{braverman2015}
V. Braverman, R. Ostrovsky, and A. Roytman. \textit{Zero-one laws for sliding windows and universal sketches}, in Proceedings of APPROX-RANDOM, pages 573--590, 2015.

\bibitem{cormode2008}
G. Cormode and M. Hadjieleftheriou. \textit{Finding frequent items in data streams}, in International Conference on Very Large Data Bases, 2008.

\bibitem{chan2005}
T. M. Chan and E. Y. Chen. \textit{Multi-pass geometric algorithms}, in Symposium on Computational Geometry, pages 180--189, 2005.

\bibitem{charikar2004}
M. Charikar, K. Chen, and M. Farach-Colton. \textit{Finding frequent items in data streams}, in Theoretical Computer Science, 314(1):3--15, 2004.

\bibitem{crouch2013}
M. S. Crouch, A. McGregor, and D. Stubbs. \textit{Dynamic graphs in the sliding-window model}, in Proceedings of the 21st Annual European Symposium on Algorithms (ESA), pages 337--348, 2013.

\bibitem{datar2002} 
M. Datar, A. Gionis, P. Indyk, and R. Motwani, \textit{Maintaining Stream Statistics over Sliding Windows}, Proc. ACM Symp. Discrete Algorithms, 2002.

\bibitem{demaine2014}
E. D. Demaine, P. Indyk, S. Mahabadi, and A. Vakilian. \textit{On streaming and communication complexity of the set cover problem}, in Distributed Computing - 28th International Symposium, DISC 2014, Austin, TX, USA, October 12-15, 2014.
Proceedings, pages 484--498, 2014.

\bibitem{drineas2003}
P. Drineas and R. Kannan. \textit{Pass efficient algorithms for approximating large matrices}, in Proc. 14th ACM-SIAM Symposium Discrete Algorithms, pages 223--232, 2003.

\bibitem{feigenbaum2005}
J. Feigenbaum, S. Kannan, A. McGregor, S. Suri, and J. Zhang. \textit{On graph problems in a semi-streaming model}, in Theor. Comput. Sci., 348(2):207--216, 2005.

\bibitem{francois2014}
N. Francois, R. Jain, and F. Magniez. \textit{Unidirectional Input/Output Streaming Complexity of Reversal and Sorting}, in Approximation, Randomization, and Combinatorial Optimization. Algorithms and Techniques (2014), p. 654.

\bibitem{greenwald2001}
J. M. Greenwald and S. Khanna. \textit{Space-efficient online computation of quantile summaries}, in Proc. the 20th ACM SIGMOD Intl. Conf. on Management of Data (SIGMOD), 2001.

\bibitem{guha2000}
S. Guha, N. Mishra, R. Motwani, and L. O'Callaghan. \textit{Clustering data streams}, in Proceedings of the 41st Annual IEEE Symposium on Foundations of Computer Science, IEEE Computer Society, Los Alamitos, CA, 2000, pp. 359--366.

\bibitem{guha2008}
S. Guha and A. McGregor. \textit{Tight lower bounds for multi-pass stream computation via pass elimination}, in Proc. of ICALP, pages 760--772. Springer, 2008.

\bibitem{henziger1998}
M. R. Henzinger, P. Raghavan, and S. Rajagopalan. \textit{Computing on Data Streams}, in Technical Report TR 1998-011, Compaq Systems Research Center, Palo Alto, CA, 1998.

\bibitem{kane2010}
D. M. Kane, J. Nelson, and D. P. Woodruff. \textit{On the exact space complexity of sketching and streaming small norms}, in Proc. SODA, pages 1161--1178, 2010.

\bibitem{mcgregor2014}
A. McGregor. \textit{Graph stream algorithms: A survey}, in SIGMOD Rec. 43(1), 2014.

\bibitem{munro1980}
J. I. Munro and M. S. Paterson. \textit{Selection and
Sorting with Limited Storage}, in Theoretical Computer
Science, Vol. 12, pages 315--323, 1980.

\bibitem{muthukrishnan2003}
S. Muthukrishnan. \textit{Data streams: algorithms and applications}, in ACM–SIAM Symposium on Discrete Algorithms, 2003, http://athos.rutgers.edu/~muthu/stream-1-1.ps.
\end{thebibliography}
\end{document}